\documentclass[a4paper,USenglish,cleveref,autoref,numberwithinsect]{lipics-v2019}
\newtheorem{conjecture}[theorem]{Conjecture}

\newcommand{\sat} {\textsc{Sat}}
\newcommand{\SAT} {\textsc{Sat}}
\newcommand{\ethh} {\mathrm{ETH}}
\newcommand{\poly} {\mathrm{poly}}

\newcommand{\cG}{{\mathcal{G}}}
\newcommand{\cP}{{\mathcal{P}}}
\newcommand{\cC}{{\mathcal{C}}}
\newcommand{\cS}{{\mathcal{S}}}

\newcommand{\scsc}{\begin{math}{\cal SC}\end{math}}
\newcommand{\sclc}{\begin{math}{\cal SC_{\cG}}\end{math}}
\newcommand{\scprime}{\begin{math}{\cal SC_{\cG'}}\end{math}}

\newcommand{\setcover}{\textsc{Set-Cover}}
\newcommand{\lc} {\textsc{Label-Cover}}
\newcommand{\dst}{\textsc{Dst}}

\hideLIPIcs
\nolinenumbers

\author{Marek Cygan}{University of Warsaw, Poland}{mcygan@mimuw.edu.pl}{}{}

\author{Magnús M. Halldórsson}{ICE-TCS \& Department of Computer Science, Reykjavik University, Iceland}{mmh@ru.is}{https://orcid.org/0000-0002-5774-8437}{}

\author{Guy Kortsarz}{Rutgers University, Camden, NJ, USA}{guyk@rutgers.edu}{}{}
\funding{The work of Marek Cygan is part of a project TOTAL that has received funding from the European Research Council (ERC) under the European Union Horizon 2020 research and innovation programme (grant agreement No 677651). Magn\'us Halld\'orsson is partially supported by Icelandic Research Fund grant 174484-051. Guy Kortsarz is partially supported by NSF grants 1540547 and 1910565.}

\authorrunning{M.\,M. Halldórsson and G. Kortsarz}

\Copyright{Marek Cygan, Magnús M. Halldórsson and Guy Kortsarz}



\acknowledgements{We thank Bundit Laekhanukit for helpful comments and discussions.}

\EventEditors{Christos Kaklamanis and Asaf Levin}
\EventNoEds{1}
\EventLongTitle{WAOA 2020}
\EventShortTitle{WAOA 2020}
\EventAcronym{WAOA}
\EventYear{2020}
\EventDate{9-10 September 2020}
\EventLocation{Pisa, Italy}
\EventLogo{}
\SeriesVolume{xxx}
\ArticleNo{xx}

\title{Tight Bounds on Subexponential Time  Approximation of Set Cover and Related Problems}
\begin{document}
\maketitle

\begin{abstract}
We show that {\setcover} on instances with $N$ elements cannot be approximated within $(1-\gamma)\ln N$-factor in time exp($N^{\gamma-\delta})$, for any $0 < \gamma < 1$ and any $\delta > 0$, assuming the Exponential Time Hypothesis. 
This essentially matches the best upper bound known by Cygan et al.\ \cite{Marek} of $(1-\gamma)\ln N$-factor in time $exp(O(N^\gamma))$.

The lower bound is obtained by extracting a standalone reduction from {\lc} to {\setcover} from the work of Moshkovitz \cite{Moshkovitz15}, and applying it to a different PCP theorem than done there. We also obtain a tighter lower bound when conditioning on the Projection Games Conjecture.

We also treat three problems (Directed Steiner Tree, Submodular Cover, and Connected Polymatroid) that strictly generalize {\setcover}.
We give a $(1-\gamma)\ln N$-approximation algorithm for these problems that runs in $exp(\tilde{O}(N^\gamma))$ time, for any $1/2 \le \gamma < 1$.
\end{abstract}

\section{Introduction}

% Set Cover, and our result
We show that {\setcover} on instances with $N$ elements cannot be approximated within $(1-\gamma)\ln N$-factor in time exp($N^{\gamma-\delta})$, assuming the Exponential Time Hypothesis ({$\ethh$}), for any $\gamma, \delta > 0$. This essentially matches the best upper bound known by Cygan et al.\ \cite{Marek}.
This is obtained by extracting a standalone reduction from {\lc} to {\setcover} from the work of Moshkovitz \cite{Moshkovitz15}, and applying it to a different PCP theorem than done there. We also obtain a tighter lower bound when conditioning on the Projection Games Conjecture.
%culminating in a time lower bound of $exp(\tilde{\Omega}(N^\gamma))$.
%We achieve this by judiciously combining known results. While all the technical components were already available, their formulation did not make it easy to infer such a result, especially regarding sub-exponential time algorithm.

% DST result
We also treat the Directed Steiner Tree (\dst) problem that strictly generalizes {\setcover}.
The input to {\dst} consists of a directed graph $G$ with costs on edges, a set of terminals, and a designated root $r$. The goal is to find a subgraph of $G$ that forms an arborescence rooted at $r$ containing all the $N$ terminals and minimizing the cost.
We give a $(1-\gamma)\ln N$-approximation algorithm for {\dst} that runs in $exp(\tilde{O}(N^\gamma))$ time, for any $\gamma \ge 1/2$.
Recall that the $\tilde{O}$-notation hides logarithmic factors.

This algorithm also applies to two other generalizations of {\setcover}. 
In the \textsc{Submodular Cover} problem, the input is a set system $(U,\cC$) with a cost on each element of the universe $U$. We are given a non-decreasing \emph{submodular} function $f:2^U\rightarrow \mathbb{R}$ satisfying,
for every $S\subseteq T\subseteq V$
and for every $x\in U\setminus T$, $f(S+x)-f(S)\geq f(T+x)-f(T)$.
The objective is to minimize the cost $c(S) = \sum_{s \in S} c(s)$ subject to $f(S) = f(U)$. 
In the \textsc{Connected Polymatroid}  problem, which generalizes \textsc{Submodular Cover}, the elements of $U$ are leaves of a tree and both elements and sets have cost. The goal is to select a set $S \subseteq U$ so that $f(S) = f(U)$ and $c(S) + c(T(S))$ is minimized, where $T(S)$ is the unique tree rooted at $r$ spanning $S$.

\subsection{Related work}
% Set cover approximation algorithms
Johnson \cite{dj} and Lov\'asz \cite{L} showed that a greedy algorithm yields a $1 + \lg N$-approximation of {\setcover}, where $N$ is the number of elements. Chv\'atal \cite{chva} extended it to the weighted version. Slav\'ik \cite{Slavik96} refined the bound to $\ln n - \ln\ln n + O(1)$.

% Hardness results (for polytime)
Lund and Yannakakis \cite{LLYY} showed that logarithmic factor was essentially best possible for polynomial-time algorithms, unless $NP \subseteq DTIME(n^{polylog(n)})$. Feige \cite{Feige} gave the precise lower bound that {\setcover} admits no $(1-\epsilon)$-approximation, for any $\epsilon > 0$, with a similar complexity assumption. Assuming the stronger {$\ethh$}, he shows that $\ln N - c\log\log N$-approximation is not possible for polynomial algorithms, for some $c > 0$. The work of Moshkovitz \cite{Moshkovitz15} and Dinur and Steurer \cite{DinurSteurer} combined shows that $(1-\epsilon)\ln N$-approximation is hard modulo $P = NP$. All the inapproximability results relate to two prover interactive proofs, via the related {\lc} problem.

% Subexponential-time results
Recent years have seen increased interest in subexponential time algorithms, including approximation algorithms. A case in point is the \emph{maximum clique} problem that has a trivial $2^{n/\alpha}\poly(n)$-time algorithm that gives a $\alpha$-approximation, for any $1 \le \alpha \le \sqrt{n}$, and Bansal et al.~\cite{bans} improved the time to $exp(n/\Omega(\alpha \log^2 \alpha))$. Chalermsook et al.~\cite{focs13} showed that this is nearly tight, as $\alpha$-approximation requires $exp(n^{1-\epsilon}/\alpha^{1+\epsilon})$ time, for any $\epsilon > 0$, assuming {$\ethh$}. 

% Set Cover
For {\setcover}, 
Cygan, Kowalik and Wykurz\cite{Marek} gave a $(1-\alpha)\ln N$-approximation algorithm that runs in time $2^{O(N^\alpha)}$, for any $0 < \alpha < 1$. 
The results of \cite{Moshkovitz15,DinurSteurer} imply a $exp(N^{\alpha/c})$-time lower bound for $(1-\alpha)\ln N$-approximation, for some constant $c \ge 3$.
An unpublished report contains a conditional $exp(\Omega(N^{\alpha}))$-time lower bound for $(1-\alpha)\ln N$-approximation of {\setcover} \cite{CKL18}. In addition to $\ethh$, this also requires the less established Projection Games Conjecture (PGC). Unfortunately, the writeup of \cite{CKL18} defies easy verification. The current paper arose from an effort to make it comprehensible.

% Other problems
\textsc{Dst} can be approximated within $N^{\epsilon}$-factor, for any $\epsilon > 0$, in polynomial time \cite{slt,moses1}. In quasi-polynomial time, it can be approximated within $O(\log^2 N / \log\log N)$-factor \cite{GLL19}, which is also best possible in that regime \cite{eran1,GLL19}. This was recently extended to \textsc{Connected Polymatroid} \cite{GN20}.
For \textsc{Submodular Cover}, the greedy algorithm also achieves a $1+\ln N$-approximation \cite{W}.
%% Mention the ln n - lglg N-approximation?

\subsection{Organization}
The paper is organized so as to be accessible at different levels of detail. In Sec.~\ref{sec:hardness}, we derive
two different hardness results for subexponential time algorithms under different complexity assumptions: ETH and PGC. For this purpose, we only state (but not prove) the hardness reduction, and introduce the {\lc} problem with its key parameters: size, alphabet size, and degrees.

We prove the properties of the hardness reduction in Sec.~\ref{sec:reduction}, by combining two lemmas extracted from \cite{Moshkovitz15}. The proofs of these lemmas are given in Sec.~\ref{sec:pf1} and \ref{sec:pf2}. To make it easy for the reader to spot the  differences with the arguments of \cite{Moshkovitz15}, we underline the changed conditions or parameters in our presentation.

Finally, the approximation algorithm for Directed Steiner Tree is given in Sec.~\ref{sec:dst}

\section{Hardness of Set Cover}
\label{sec:hardness}

We give our technical results in this section. Starting with definition of {\lc} in Sec.~\ref{sec:prelim}, we give a reduction from {\lc} to {\setcover} in Sec.~\ref{sec:reduction}, and derive specific approximation hardness results in Sec.~\ref{sec:results}. A full proof of the correctness of the reduction is given in the following section.

%\subsection{Preliminaries}
\subsection{Label Cover}
\label{sec:prelim}

The intermediate problem in all known approximation hardness reductions for {\setcover} is the {\lc} problem.
%\paragraph{Label Cover}
%% Definition of Label Cover
\begin{definition}
In the $\lc$ problem with the projection property (a.k.a., the \emph{Projection Game}),
we are given a bipartite graph $G(A,B,E)$, finite alphabets (also called \emph{labels}) $\Sigma_A$ 
and $\Sigma_B$, and a function $\pi_e: \Sigma_A\rightarrow \Sigma_B$ for each edge $e\in E$.
A \emph{labeling} is a pair  
$\varphi_A:A\rightarrow \Sigma_A$ and $\varphi_B:B\rightarrow \Sigma_B$ 
of assignments of labels to the vertices of $A$ and $B$, respectively.
An edge $e=(a,b)$ is \emph{covered} (or \emph{satisfied}) by $(\varphi_A,\varphi_B)$ if
$\pi_e(\varphi_A(a))=\varphi_B(b)$.
The goal in $\lc$ is to find a labeling $(\varphi_A,\varphi_B)$
that covers as many edges as possible.
\end{definition}

%% Parameters of Label Cover
The \emph{size} of a label cover instance $\cG = (G=(A,B,E),\Sigma_A, \Sigma_B, \Pi = \{\pi_e\}_e)$
is denoted by $n_{\cG} = |A| + |B| + |E|$. 
The \emph{alphabet size} is $\max(|\Sigma_A|, |\Sigma_B|)$.
The {\lc} instances we deal with will be \emph{bi-regular}, meaning that all nodes of the same bipartition have the same degree. We refer to the degree of nodes in $A$ ($B$) as the \emph{$A$-degree} (\emph{$B$-degree}), respectively.

%% Define soundness error
{\lc} is a central problem in computational complexity, corresponding to \emph{projection PCPs}, or  probabilistically checkable proofs that make 2 queries.
\iffalse  %% Start of PCP comment
\footnote{In the context of the \emph{Two-Provers One-Round game} (2P1R), 
every label is an answer to some "question" $a$ sent to
Player $A$ and some question $b$ sent to Player $B$,
for a query $(a,b)\in E$. The two answers make the verifier accept
if a label $x\in \Sigma_x$ assigned to $a$ and 
a label $y\in \Sigma_B$ assigned to $b$
satisfy $\phi(x)=y$. Since any label $x \in \Sigma_A$ has a unique label in $\Sigma_B$ 
that causes the verifier to accept, $y$ is called
the \emph{projection of $x$ into $b$}.
We will not elaborate on the relationship to two-prover interactive proofs in this paper.}
\fi %% End of PCP comment
A key parameter is the \emph{soundness error}:
\begin{definition}
  A {\lc} construction $\cG=\cG_\phi$, formed from a 3-SAT formula $\phi$, has \emph{soundness error} $\epsilon$ if: a) whenever $\phi$ is satisfiable, there is a labeling of $\cG$ that covers all edges, and b) when $\phi$ is unsatisfiable, every labeling of $\cG$ covers at most an $\epsilon$-fraction of the edges.
\end{definition}
  
Note that the construction is a schema that holds for a range of values for $\epsilon$.
A {\lc} construction is \emph{almost-linear size} if it is of size $n^{1+o(1)}$, possibly with extra $poly(1/\epsilon)$ factors.

  We use the following PCP theorem of Moshkovitz and Raz  \cite{r2}.
\begin{theorem}[\cite{r2}]
%There exists $c > 0$ such that for every $\epsilon \ge 1/n^c$,
For every $\epsilon \ge 1/polylog(n)$, {\sat} on input of size $n$ can be reduced to {\lc} on a bi-regular graph of degrees $\poly(1/\epsilon)$, with soundness error $\epsilon$, size $n^{1+o(1)}\poly(1/\epsilon) = n^{1+o(1)}$, and alphabet size that is exponential in $poly(1/\epsilon)$. The reduction can be computed in time linear in the size and alphabet size of the Label-Cover.
\label{thm:RazMos}
\end{theorem}

Dinur and Steurer \cite{DinurSteurer} later gave a PCP construction whose alphabet size depends only polynomially on $1/\epsilon$. This is crucial for NP-hardness results, and combined with the reduction of \cite{Moshkovitz15}, implies the essentially tight bound of $(1-\epsilon)\ln N$-approximation of {\setcover} by poly-time algorithms. 

\subsection{Set Cover Reduction}
\label{sec:reduction}

We present here a reduction from a generic {\lc} (a two-prover PCP theorem) to the {\setcover} problem. This is extracted from the work of Moshkovitz \cite{Moshkovitz15}. The presentation in \cite{Moshkovitz15} was tightly linked with the PCP construction of Dinur and Steurer \cite{DinurSteurer} that was used in order to stay within polynomial time. When allowing superpolynomial time, it turns out to be more frugal to apply the older PCP construction of Moshkovitz and Raz \cite{r2}. This construction has exponential dependence on the alphabet size, which precludes its use in NP-hardness results. On the other hand, it has nearly-linear dependence on the size of the Label Cover, unlike Dinur-Steurer, and this becomes a dominating factor in subexponential reductions. 

Our main technical contribution is then to provide a standalone reduction from {\lc} to {\setcover} that allows specific PCP theorems to be plugged in.

We say that a reduction that originates in {\sat} achieves \emph{approximation gap} $\rho$ if there is a value $a$ such that {\setcover} instances originating in satisfiable formulas have a set cover of size at most $a$, while instances originating in unsatisfiable formulas have all set covers of size greater than $\rho \cdot a$.

\begin{theorem}
Let $\gamma > 0$ and $0 < \delta < \gamma$. 
There is a reduction from {\lc} to {\setcover} with the following properties. Let $\cG$ be a bi-regular {\lc} of almost-linear size $n_0$, 
soundness error parameter $\epsilon$, $B$-degree $poly(1/\epsilon)$, and alphabet size $\sigma_A(\epsilon)$. 
%\
Then for each $\gamma > 0$, $\cG$ is reduced to a {\setcover} instance {\scsc}=\scsc$_{\cG,\gamma, \delta}$  with approximation gap $(1-\gamma)\ln N$,
% \sc$_{\cG,\gamma, \delta}$
$N = \tilde{O}(n_0^{1/(\gamma-\delta)})$ elements, and $M = \tilde{O}(n_0) \cdot \sigma_A(polylog(n))$ sets.
The time of the reduction is linear in the size of {\scsc}.
\label{T:setcover}
\end{theorem}

\subsection{Approximation Hardness Results}
\label{sec:results}

When it comes to hardness results for subexponential time algorithms the standard assumption is the \emph{Exponential Time Hypothesis (ETH)}.
%\paragraph{ETH}
ETH asserts that the 3-$\sat$ problem on $n$ variables and $m$ clauses cannot be solved in $2^{o(n)}$-time. Impagliazzo, Paturi and Zane \cite{eth} showed that any 3-$\sat$ instance can be sparsified in $2^{o(n)}$-time to an instance with $m=O(n)$ clauses. 
When we refer to ${\SAT}$ input of size $n$, we mean 3-CNF formula on $n$ variables and $O(n)$ clauses.
Thus, $\ethh$ together with the sparsification lemma \cite{ethm} implies the following:
\begin{conjecture}
($\ethh$)
%Given a boolean SAT input $\phi$ of size $n$, there is no $2^{o(n)}$-time algorithm that decides whether $\phi$ is satisfiable.
There is no $2^{o(n)}$-time algorithm that decides {\SAT} on inputs of size $n$. %Given a boolean SAT input $\phi$ of size $n$, there is no $2^{o(n)}$-time algorithm that decides whether $\phi$ is satisfiable.
\end{conjecture}
%{\bf Exponential-Time Hypothesis combined with the Sparsification Lemma:}
% In particular, 3-SAT admits no subexponential-time algorithm.
We need only a weaker version: There is some $\zeta > 0$ such that there is no $exp(n^{1-\zeta})$-time algorithm to decide {\SAT}.

We are now ready for our main result.

\begin{theorem}
  Let $0 < \gamma < 1$ and $0 < \delta < \gamma$.
  Assuming ETH, there is no $(1-\gamma)\ln N$-approximation algorithm of {\setcover} with $N$ elements and $M$ sets that runs in time $exp(N^{\gamma - \delta}) \cdot poly(M)$.
\label{thm:mainresult}
\end{theorem}

\begin{proof}
We show how such an algorithm can be used to decide {\sat} in subexponential time, contradicting ETH.
Given a {\sat} instance of size $n$, apply  
Thm.~\ref{thm:RazMos} with $\epsilon = \Theta(1/\log^2 n)$, to obtain {\lc} instance $\cG$ of size $n_0 = n^{1+o(1)}$ and
alphabet size $\sigma_A(\epsilon) = exp(poly(1/\epsilon)) = exp(polylog(n))$.
%and $B$-degree $q \doteq \poly(1/\epsilon) = polylog(n)$.
Next apply Thm.~\ref{T:setcover} to obtain a {\setcover} instance {\scsc=\sclc} with $M = \tilde{O}(n_0) \sigma_A(polylog(n)) = exp(polylog(n))$ sets and $N = \tilde{O}(n_0^{1/\gamma}) = n^{1/\gamma+o(1)}$ elements, and approximation gap $(1-\gamma)\ln N$.

Suppose there is a $(1-\gamma)\ln N$-approximation algorithm of {\setcover} running in time $exp(N^{\gamma-\delta}) \cdot poly(M)$. 
Since it achieves this approximation, it can decide the satisfiability of $\phi$.
Since $N^{\gamma-\delta} = n^{(1/\gamma+o(1))\cdot (\gamma-\delta)} \le n^{1 - \delta'}$, for some $\delta' > 0$, the running time contradicts ETH. 
\end{proof}

\textbf{Note}: We could also allow the algorithm greater than polynomial complexity in terms of $M$ without changing the implication. In fact, the complexity can be as high as $exp(M^{\delta_0})$, for some small constant $\delta_0$.
%or even $exp(exp(\log^\kappa M))$, for a small enough constant $\kappa$.

\subsubsection{Still Tighter Bound Under Stronger Assumptions}
Moshkovitz \cite{Moshkovitz15} proposed a conjecture on the parameters of possible Label Cover constructions.
We require a particular version with almost linear size and low degree. 
\begin{definition}[The Projection Games Conjecture (PGC)]
3-SAT of inputs of size $n$ can be reduced to {\lc}
 of size $n^{1+o(1)}$, alphabet size $poly(1/\epsilon)$, and bi-regular degrees $poly(1/\epsilon)$, where $\epsilon$ is the soundness error parameter.
\end{definition}

The key difference of PGC from known PCP theorems is the alphabet size.
PGC is considered quite plausible and has been used to prove conditional hardness results for a number of problems \cite{Moshkovitz15,Muk19}.

By assuming PGC, we improve the dependence on $M$, the number of sets.
%The $o(1)$ in the statement below means that there exists a function of growth $o(1)$ for which the claim holds.

\begin{theorem}
  Let $0 < \gamma < 1$ and $0 < \delta < \gamma$.
  Assuming PGC and ETH, there is no $(1-\gamma)\ln N$-approximation, nor a $O(\log M)$-approximation, of {\setcover} with $N$ elements and $M$ sets that runs in time $exp(N^{\gamma - \delta} M^{1-\delta})$.
%  $2^{N^{\gamma - o(1)} M^{1-o(1)} }$.
  \label{thm:pgc-result}
\end{theorem}

Namely, both the approximation factor and the time complexity can depend more strongly on the number of sets in the {\setcover} instance. The only result known in terms of $M$ is a folklore $\sqrt{M}$-approximation in polynomial time. 
%(TO DO)

\begin{proof}
We can proceed in the same way as in the proof of Thm.~\ref{thm:mainresult}, but starting from the conjectured {\lc} given by PGC, in which the alphabet size is polynomial in $\epsilon$. We then obtain a set cover instance {\sclc} that differs only in that we now have $M = |A||\Sigma_A| = n^{1+o(1)} = N^{\gamma+o(1)}$. So, a $c\log M$-approximation, with constant $c > 0$, implies a $c\gamma \ln N$-approximation, which is smaller than $(1-\gamma)\ln N$-approximation when $c\gamma < (1-\gamma)$ or $\gamma < 1/(2c)$. 
Also, $exp(M^{1-\delta}) = exp(n^{1-\delta'})$, for some $\delta' > 0$, and hence such a running time again breaks ETH.
\end{proof}

\section{Proof of the Set Cover Reduction}
\label{sec:reduction-proof}

We extract here a sequence of two reductions 
from the work of Moshkovitz. By untangling them from the {\lc} construction, we can use them for our standalone {\setcover} reduction. 

Moshkovitz \cite{Moshkovitz15} chooses some of the parameters of the lemmas so as to fit the purpose of proving NP-hardness of approximation. As a result, the size of the intermediate {\lc} instance generated grows to be a polynomial of degree larger than 1. This leads to weaker hardness results for sub-exponential time algorithms than what we desire. 
We indicate therefore how we can separate a key parameter to maintain nearly-linear size label covers.

%Moshkovitz proved her result via two reduction steps. 
A key tool in her argument is the concept of agreement soundness error.

\begin{definition}[List-agreement soundness error]
Let $\cG$ be a {\lc} for deciding the satisfiability of a Boolean formula $\phi$. Let $\varphi_A$ assign each $A$-vertex $\ell$ alphabet symbols. We say that the A-vertices \emph{totally disagree} on a vertex $b \in B$ if, there are no two neighbors $a_1, a_2 \in A$ of $b$ for which there exist $\sigma_1 \in \varphi_A(a_1), \sigma_2 \in \varphi_A(a_2)$ such that $\pi_{(a_1,b)}(\sigma_1) = \pi_{(a_2,b)}(\sigma_2)$.

We say that $\cG$ has \emph{list-agreement soundness error} $(\ell,\epsilon)$ if, for unsatisfiable $\phi$, for any assignment $\varphi_A:A \rightarrow {\binom{\sigma_A}{\ell}}$, the A-vertices are in total disagreement on at least $1-\epsilon$ fraction of the $B$-vertices.
\end{definition}

The reduction of Thm.~\ref{T:setcover} is obtained by stringing together two reductions: from {\lc} to a modified {\lc} with a low agreement soundness error, and from that to {\setcover}.

The first one is laid out in the following lemma that combines Lemmas 4.4 and 4.7 of \cite{Moshkovitz15}.
The proof is given in the upcoming subsection.

\begin{lemma}
Let $D \ge 2$ be a prime power, $q$ be a power of $D$, $\ell > 1$, and $\epsilon_0 > 0$. 
There is a polynomial reduction from a {\lc} with soundness error $\epsilon_0^2 D^2$ and $B$-degree $q$ to a {\lc} with list-agreement soundness error $(\ell, 2\epsilon_0 D^2 \cdot \ell^2)$ and $B$-degree $D$. The reduction preserves alphabets, and the size is increased by $poly(q)$-factor. 
\label{l:mos-combined}
\end{lemma}

Moshkovitz also gave a reduction from {\lc} with small agreement soundness error to {\setcover} approximation.
We extract a more general parameterization than is stated explicitly around Claim 4.10 in \cite{Moshkovitz15}. The proof, with minor changes from \cite{Moshkovitz15}, is given for completeness in Sec.~\ref{sec:pf2}.

\begin{lemma}[\cite{Moshkovitz15}, rephrased]
Let $\cG' = (G'=(A',B',E'), \Sigma_A, \Sigma_B, \Pi')$ be a bi-regular {\lc} instance with soundness parameter $\epsilon$  for deciding the satisfiability of a boolean formula $\phi$. Let $D$ be the $B'$-degree.
For every $\alpha$ \underline{with $2/D < \alpha < 1$},\footnote{Moshkovitz doesn't explicitly relate $\alpha$ and $D$, but indicates that $\alpha$ be small and $D$ "sufficiently large".} and any 
$u \ge (D^{O(\log D)} \log |\Sigma_{B'}|)^{1/\alpha}$,
%$u \ge (D^{c_0 \log D} \log |\Sigma_{B'}|)^{1/\alpha}$ for a certain absolute constant $c_0$,
there is a reduction from $\cG$ to a {\setcover} instance {\scsc=\scprime} with a certain choice of $\epsilon$ that attains the following properties: 
\begin{enumerate}
  \setlength{\itemsep}{1pt}
  \setlength{\parskip}{0pt}
  \setlength{\parsep}{0pt}
    \item \emph{Completeness:} If all edges of $G'$ can be covered, then {\scsc} has a set cover of size $|A'|$.
    \item \emph{Soundness:} If $\cG'$ has list-agreement soundness error $(\ell, \alpha)$, where $\ell = D (1-\alpha) \ln u $,
then every set cover of {\scsc} is of size more than $|A'|(1-2\alpha)\ln u$. 
    \item The number $N$ of elements of {\scsc} is $|B'| \cdot u$ and the number $M$ of sets is $|A'| \cdot |\Sigma_A|$. 
    \item The time for the reduction is polynomial in $|A'|, |B'|, |\Sigma_A|, |\Sigma_B|$ and $u$. 
\end{enumerate}
\label{lem:setcover}
\end{lemma}

Given these lemmas, we can now prove Thm.~\ref{T:setcover}, which we restate for convenience.

\noindent \textbf{Theorem \ref{T:setcover} (restated)}. 
\emph{
Let $\gamma > 0$ and $0 < \delta < \gamma$. 
There is a reduction from {\lc} to {\setcover} with the following properties. Let $\cG$ be a bi-regular {\lc} of almost-linear size $n_0$, 
soundness error parameter $\epsilon$, $B$-degree $poly(1/\epsilon)$, and alphabet size $\sigma_A(\epsilon)$. 
%\
Then for each $\gamma > 0$, $\cG$ is reduced to a {\setcover} instance {\scsc}=\scsc$_{\cG,\gamma, \delta}$  with approximation gap $(1-\gamma)\ln N$,
% \sc$_{\cG,\gamma, \delta}$
$N = \tilde{O}(n_0^{1/(\gamma-\delta)})$ elements, and $M = \tilde{O}(n_0) \cdot \sigma_A(polylog(n))$ sets.
The time of the reduction is linear in the size of {\scsc}.
}

%\begin{proof}[Proof of Thm.~\ref{T:setcover}]
\begin{proof}[Proof of Thm.~\ref{T:setcover}]
Let $\alpha = 2\delta$, $D$ be a prime power at least $2/\alpha$, and $\gamma' = (\gamma-\delta)/(1-\delta)$.
Let $n_1 = \tilde{O}(n_0)$ be the size of the instance that is formed by Lemma \ref{l:mos-combined} on $\cG$ with $\epsilon = \Theta(1/\log^4 n)$.
Let $u = n_1^{(1-\gamma')/\gamma'}$ and note that $\log u = \Theta(\log n)$.
Let $\ell = D(1-\alpha) \ln u = \Theta(\log n)$, 
$\epsilon = \alpha^2 D^{-2} \ell^{-4}/4 = \Theta(1/\log^4 n)$, and $\epsilon_0 = \sqrt{\epsilon}/D = \alpha/(2D^2\ell^2) = \Theta(1/\log^2 n)$.
Let $q$ be the $B$-degree of $\cG$.

We apply Lemma \ref{l:mos-combined} to $\cG$ (with $D$, $q$, $\ell$ and $\epsilon_0$) and obtain a {\lc} instance $\cG' = (G'=(A',B',E'), \Sigma_A, \Sigma_B, \Pi)$ with $B$-degree $D$ and size $n_1$, with alphabets unchanged.
The list-agreement soundness error of $\cG'$ is $(\ell, 2\epsilon_0 \ell^2 D^2)$, or $(\ell, \alpha)$. 
%Then, the size $n_1 = n^{1+o(1)}$ and alphabet size $exp(polylog(n))$.

We can verify that $\cG'$, $\alpha$, and $D$ and $u$ satisfy the prerequisites of Lemma \ref{lem:setcover}, which yields a {\setcover} instance {\scsc=\sclc} with $N = |B'| \cdot u = n_1 \cdot n_1^{(1-\gamma')/\gamma'} = n_1^{1/\gamma'} = \tilde{O}(n_0^{(1-\delta)/(\gamma-\delta)}) = \tilde{O}(n_0^{1/(\gamma-\delta)})$ elements and $M = |A'| \cdot |\Sigma_A|$ sets.
%Observe that $N = n_1 \cdot n_0^{(1-\gamma)/\gamma}$, while the value of $M$ depends on the alphabet size.

If $\phi$ is satisfiable, then {\scsc} has set cover of size $|A'|$, while if it is unsatisfiable, then every set cover of {\scsc} has size more than $|A'|(1-2\alpha)\ln u = |A'|(1-\delta)\ln u$. 
Note that $\ln u = \ln (N/n_1) = (1-\gamma') \ln N$.
Hence, the approximation gap is $(1-\delta)\ln u = (1-\delta)(1-\gamma')\ln N = (1-\gamma)\ln N$.
\end{proof}

%\section{Proofs of Lemmas}

\subsection{Proof of Lemma \ref{l:mos-combined}}
\label{sec:pf1}

We give here a full proof of Lemma \ref{l:mos-combined}, based on \cite{Moshkovitz15}, with minor modification.

When the labeling assigns a single label to each node, i.e., when $\ell=1$, the list-agreement soundness error reduces to \emph{agreement soundness error}, which is otherwise defined equivalently.
Moshkovitz first showed how to reduce a {\lc} with small soundness error to one with a small agreement soundness error. The lemma stated here is unchanged from \cite{Moshkovitz15} except that Moshkovitz used the parameter name $n$ instead of our parameter $q$. \footnote{This invited confusion, since $n$ was also used to denote the size of the {\lc} (like we do here).}

\begin{lemma}[Lemma 4.4 of \cite{Moshkovitz15}]
Let $D \ge 2$ be a prime power and \underline{let $q$ be a power of $D$}. Let $\epsilon_0 > 0$. 
There is a polynomial reduction from a {\lc} with soundness error $\epsilon_0^2 D^2$ and $B$-degree $q$ to a {\lc} with agreement soundness error $2\epsilon_0 D^2$ and $B$-degree $D$. The reduction preserves alphabets, \underline{and the size is increased by $poly(q)$-factor.} 
\label{l:mos}
\end{lemma}

We have underlined the parts that changed because of using $q$ as parameter instead of $n$.
The proof is based on the following combinatorial lemma, whose proof we omit. 
We note that the set $U$ here is different from the one used in Lemma \ref{lem:setcover} (but we retained the notation to remain faithful to \cite{Moshkovitz15}).

\begin{lemma}[Lemma 4.3 of \cite{Moshkovitz15}]
For $0 < \epsilon < 1$, for a prime power $D$, and $q$ that is a power of $D$, there is an explicit construction of a regular bipartite graph $H=(U,V,E)$ with $|U| = q$, $V$-degree $D$, and $|V| \le q^{O(1)}$ that satisfies the following. For every partition $U_1, \ldots, U_\ell$ of $U$ into sets such that $|U_i| \le \epsilon |U|$ for $i=1,2,\ldots, \ell$, the fraction of vertices $v \in V$ with more than one neighbor in any single set $U_i$, is at most $\epsilon D^2$.
\label{lem:combinatorial}
\end{lemma}

Again, we used the parameter name $q$, rather than $n$ as in \cite{Moshkovitz15}.
We show how to take a {\lc} with standard soundness and convert it to a {\lc} instance with total disagreement soundness, by combining it with the graph from Lemma \ref{lem:combinatorial}

\begin{proof}[Proof of Lemma \ref{l:mos}]
Let $\cG = (G=(A,B,E), \Sigma_A, \Sigma_B, \Pi)$ be the original {\lc}. Let $H=(U,V,E_H)$ be the graph from Lemma \ref{lem:combinatorial}, where $q, D$ and $\epsilon$ are as given in the current lemma. Let us use $U$ to enumerate the neighbors of a $B$-vertex, i.e., there is a function $E^{\leftarrow}:B \times U \rightarrow A$ that given a vertex $b \in B$ and $u \in U$ gives us the $A$-vertex which is the $u$ neighbor (in $G$) of $b$.

We create a new {\lc} $\cG' = (G=(A,B\times V, E'), \Sigma_A, \Sigma_B, \Pi')$. The intended assignment to every vertex $a \in A$ is the same as its assignment in the original instance. The intended assignment to a vertex $\langle b,v\rangle \in B \times V$ is the same as the assignment to $b$ in the original game. We put an edge $e' = (a, \langle b,v\rangle)$ if $E^{\leftarrow}(b,u) = a$ and $(u,v) \in E_H$. We define $\pi_{e'} \equiv \pi_{(a,b)}$.

If there is an assignment to the original instance that satisfies $c$ fraction of its edges, then the corresponding assignment to the new instance satisfies $c$ fraction of its edges.

Suppose there is an assignment for the new instance $\varphi_A : A \rightarrow \Sigma_A$ in which more than $2 \epsilon D^2$ fraction of the vertices in $B \times V$ do not have total disagreement.

Let us say that $b \in B$ is "good" if for more than an $\epsilon D^2$ fraction of the vertices in $\{b\} \times V$, the $A$-vertices do not totally disagree. Note that the fraction of good $b \in B$ is at least $\epsilon D^2$. 

Focus on a good $b \in B$. Consider the partition of $U$ into $|\Sigma_B|$ sets, where the set corresponding to $\sigma \in \Sigma_B$ is:
\[ U_\sigma = \{u \in U | a = E^\leftarrow(b,u) \wedge e=(a,b) \in E_G \wedge \pi_e(\varphi_A(a))=\sigma \}\ . \]

By the goodness of $b$ and the property of $H$, there must be $\sigma \in \Sigma_B$ such that $|U_\sigma| > \epsilon |U|$. We call $\sigma$ the "champion" for $b$.

We define an assignment $\varphi_B:B\rightarrow \Sigma_B$ that assigns good vertices $b$ their champions, and other vertices $b$ arbitrary values. The fraction of edges that $\varphi_A$, $\varphi_B$ satisfy in the original instance is at least $\epsilon^2 D^2$. 
\end{proof}

%\paragraph{List version}
Moshkovitz then shows that small agreement soundness error translates to the list version.
The proof is unchanged from \cite{Moshkovitz15}.

\begin{lemma}[Lemma 4.7 of \cite{Moshkovitz15}]
%\begin{lemma}[{\lc} with list-agreement soundness, Lemma 4.7 
Let $\ell \ge 1$, $0 < \epsilon' < 1$. A {\lc} with agreement soundness error $\epsilon'$ has list-agreement soundness error $(\ell, \epsilon' \ell^2)$. 
\label{l:las}
\end{lemma}

\begin{proof}
Assume by the way of contradiction that the {\lc} instance has an assignment $\hat{\varphi}_A:A \rightarrow {\binom{\sigma_A}{\ell}}$ such that on more than $\epsilon' \ell^2$-fraction of the $B$-vertices, the $A$-vertices do not totally disagree. Define an assignment $\varphi_A:A \rightarrow \Sigma_A$ by assigning every vertex $a \in A$ a symbol picked uniformly at random from the $\ell$ symbols in $\hat{\varphi}_A(a)$. If a vertex $b\in B$ has two neighbors $a_1,a_2 \in A$ that agree on $b$ under the list assignment $\hat{\varphi}_A$, then the probability that they agree on $b$ under the assignment $\varphi_A$ is at least $1/\ell^2$. Thus, under $\varphi_A$, the expected fraction of the $B$-vertices that have at least two neighbors that agree on them, is more than $\epsilon'$. In particular, there exists an assignment to the $A$-vertices, such that more than $\epsilon'$ fraction of the $B$-vertices have two neighbors that agree on them. This contradicts the agreement soundness.
\end{proof}

Lemma \ref{l:mos-combined} follows directly from combining Lemmas \ref{l:mos} and \ref{l:las}.

\subsection{Proof of Lemma \ref{lem:setcover}}
\label{sec:pf2}

We give here a proof of Lemma \ref{lem:setcover}, following closely the exposition of \cite{Moshkovitz15}, with minor modifications. 

Feige \cite{Feige} introduced the concept of \emph{partition systems} (also known as \emph{anti-universal sets} \cite{arvind}) which is key to tight inapproximbility results for {\setcover}. It consists of a universe along with a collection of partitions. Each partition covers the universe, but any cover that uses at most one set out of each partition, is necessarily large. The idea is to form the reduction so that if the {\sat} instance is satisfiable, then one can use a single partition to cover the universe, while if it is unsatisfiable, then one must use sets from different partitions, necessarily resulting in a large set cover. 

Naor, Schulman, and Srinivasan \cite{arvind} gave the following 
combinatorial construction (which as appears as Lemma 4.9 of \cite{Moshkovitz15}) that derandomizes one introduced by Feige \cite{Feige}. 

\begin{lemma}[\cite{arvind}]
For natural numbers $m$, $D$ and $0 < \alpha < 1$, \underline{with $\alpha \ge 2/D$}, and for all $u \ge (D^{O(\log D)}\log m)^{1/\alpha}$, there is an explicit construction of a universe $U$ of size $u$ and partitions $\cP_1, \ldots, \cP_m$ of $U$ into $D$ sets that satisfy the following: there is no cover of $U$ with $\ell = D \ln |U| (1-\alpha)$ sets $S_{i_1}, \ldots, S_{i_\ell}$, $1 \le i_1 < \cdots < i_\ell \le m$, such that set $S_{i_j}$ belongs to partition $\cP_{i_j}$.
\label{l:nss}
\end{lemma}

Naor et al \cite{arvind} state the result in terms of the relation $u \ge (D/(D-1))^\ell \ell^{O(\log \ell} \log m$. 
Note that for $\ell = D \ln u (1-\alpha)$, we have $(D/(D-1))^\ell \approx u^{(1-\alpha)D/(D-1)} \approx u^{1 - \alpha + 1/D}$, and hence we need $D$ to be sufficiently large.

The following reduction follows Moshkovitz \cite{Moshkovitz15}, which in turns is along the lines of Feige \cite{Feige}.

\noindent \textbf{Lemma \ref{lem:setcover} (restated).} 
\emph{
Let $\cG' = (G'=(A',B',E'), \Sigma_A, \Sigma_B, \Pi')$ be a bi-regular {\lc} instance with soundness parameter $\epsilon$  for deciding the satisfiability of a boolean formula $\phi$. Let $D$ be the $B'$-degree.
For every $\alpha$ \underline{with $2/D < \alpha < 1$},\footnote{Moshkovitz doesn't explicitly relate $\alpha$ and $D$, but indicates that $\alpha$ be small and $D$ "sufficiently large".} and any 
$u \ge (D^{O(\log D)} \log |\Sigma_{B'}|)^{1/\alpha}$,
%$u \ge (D^{c_0 \log D} \log |\Sigma_{B'}|)^{1/\alpha}$ for a certain absolute constant $c_0$,
there is a reduction from $\cG$ to a {\setcover} instance {\scsc=\scprime} with a certain choice of $\epsilon$ that attains the following properties: 
\begin{enumerate}
  \setlength{\itemsep}{1pt}
  \setlength{\parskip}{0pt}
  \setlength{\parsep}{0pt}
    \item \emph{Completeness:} If all edges of $G$ can be covered, then {\scsc} has a set cover of size $|A'|$.
    \item \emph{Soundness:} If $\cG'$ has list-agreement soundness error $(\ell, \alpha)$, where $\ell = D (1-\alpha) \ln u $,
then every set cover of {\scsc} is of size more than $|A'|(1-2\alpha)\ln u$. 
    \item The number $N$ of elements of {\scsc} is $|B'| \cdot u$ and the number $M$ of sets is $|A'| \cdot |\Sigma_A|$. 
    \item The time for the reduction is polynomial in $|A'|, |B'|, |\Sigma_A|, |\Sigma_B|$ and $u$. 
\end{enumerate}
}

\begin{proof}
%\begin{proof}[Proof of Lemma \ref{lem:setcover}]
Let $\alpha$ and $u$ be values satisfying the statement of the theorem.
Let $m = |\Sigma_B|$ and let $D$ be the $B$-degree of $\cG$. Apply Lemma \ref{l:nss} with $m$, $D$ and $u$, obtaining a universe $U$ of size $u$ and partitions $\cP_{\sigma_1}, \ldots, \cP_{\sigma_m}$ of $U$.
We index the partitions by the symbols $\sigma_1, \ldots, \sigma_m$ of $\Sigma_B$. The elements of the {\setcover} instances are $B \times U$. Equivalently, each vertex $b \in B$ has a copy of the universe $U$. Covering this universe corresponds to satisfying the edges that touch $b$. There are $m$ ways to satisfy the edges that touch $b$ --- one for every possible assignment $\sigma \in \Sigma_B$ to $b$. The different partitions covering $u$ correspond to those different assignments.

For every vertex $a \in A$ and an assignment $\sigma \in \Sigma_A$ to $a$ we have a set $S_{a,\sigma}$ in the {\setcover} instance. Taking $S_{a,\sigma}$ to the cover corresponds to assigning $\sigma$ to $a$. Notice that a cover might consist of several sets of the form $S_{a,\cdot}$ for the same $a \in A$, which is the reason we consider list agreement. The set $S_{a,\sigma}$ is a union of subsets, one for every edge $e=(a,b)$ touching $a$. If $e$ is the $i$-th edge coming into $b$ ($1\le i \le D$), then the subset associated with $e$ is $\{b\} \times S$, where $S$ is the $i$-th subset of the partition $P_{\phi_e(\sigma)}$.

%If we have an assignment to the $A$-vertices such that all of the neighbors of $b$ agree on one value for $b$, then the $D$ subsets corresponding to those neighbors and their assignments form a partition that covers $b$'s universe. On the other hand, if one uses only sets that correspond to totally disagreeing assignments to the neighbors, then by the definition of the partitions, covering $u$ requires $\approx \ln |U|$ times more sets.

% Completeness
Completeness follows from taking the set cover corresponding to each of the $A$-vertices and its satisfying assignments.

% Soundness
To prove soundness, assume by contradiction that there is a set cover $C$ of {\sclc} of size at most $|A|\ln |U| (1-2\alpha)$. For every $a \in A$, let $s_a$ be the number of sets in $C$ of the form $S_{a,\cdot}$. Hence, $\sigma_{a\in A} s_a = |C|$. For every $b \in B$, let $s_b$ be the number of sets in $C$ that participate in covering $\{b\} \times U$. Then, denoting the $A$-degree of $G$ by $D_A$,
\[ \sum_{b\in B} s_b = \sum_{a\in A} s_a D_A \le D_A |A| \ln |U|(1-2\alpha) = D|B|\ln |U| (1-2\alpha)\ . \]
In other words, on average over the $b \in B$, the universe $\{b\}\times U$ is covered by at most $D \ln |U|(1-2\alpha)$ sets. Therefore, by Markov's inequality, the fraction of $b\in B$ whose universe $\{b\}\times U$ is covered by at most $D\ln |U|(1-\alpha) = \ell$ sets is at least $\alpha$. By the contrapositive of Lemma \ref{l:nss} and our construction, for such $b \in B$, there are two edges $e_1 = (a_1, b), e_2 = (a_2,b) \in E$ with $S_{a_1,\sigma_1}, S_{a_2,\sigma_2} \in C$ where $\pi_{e_1}(\sigma_1) = \pi_{e_2}(\sigma_2)$.

We define assignment $\hat{\varphi}_A: A \rightarrow \binom{\sigma_A}{\ell}$ to the $A$-vertices as follows. For every $a \in A$, pick $\ell$ different symbols $\sigma \in \Sigma_A$ from those with $S_{a,\sigma} \in C$ (add arbitrary symbols if there are not enough). As we showed, for at least $\alpha$-fraction of the $b \in B$, the $A$-vertices will not totally disagree. Hence, the soundness property follows.
\end{proof}

\iffalse
\section{Missing Proof}

\noindent\textbf{Lemma \ref{l:las} (Lemma 4.7 of \cite{Moshkovitz15})}
\emph{Let $\ell \ge 1$, $0 < \epsilon' < 1$. A {\lc} with agreement soundness error $\epsilon'$ has list-agreement soundness error $(\ell, \epsilon' \ell^2)$. }

\begin{proof}
Assume by the way of contradiction that the {\lc} instance has an assignment $\hat{\varphi}_A:A \rightarrow {\binom{\sigma_A}{\ell}}$ such that on more than $\epsilon' \ell^2$-fraction of the $B$-vertices, the $A$-vertices do not totally disagree. Define an assignment $\varphi_A:A \rightarrow \Sigma_A$ by assigning every vertex $a \in A$ a symbol picked uniformly at random from the $\ell$ symbols in $\hat{\varphi}_A(a)$. If a vertex $b\in B$ has two neighbors $a_1,a_2 \in A$ that agree on $b$ under the list assignment $\hat{\varphi}_A$, then the probability that they agree on $b$ under the assignment $\varphi_A$ is at least $1/\ell^2$. Thus, under $\varphi_A$, the expected fraction of the $B$-vertices that have at least two neighbors that agree on them, is more than $\epsilon'$. In particular, there exists an assignment to the $A$-vertices, such that more than $\epsilon'$ fraction of the $B$-vertices have two neighbors that agree on them. This contradicts the agreement soundness.
\end{proof}
\fi

\section{Approximation Algorithm for Directed Steiner Tree}
\label{sec:dst}

Recall that in {\dst}, the input consists of a directed graph $G$ with costs $c(e)$ on edges, a collection $X$ of terminals, and a designated root $r \in V$. 
%The minimum cost $r$-rooted subtree (or \emph{arborescence}) of $G$ that contains $X$ is denoted $T(r,X)$.
The goal is to find a subgraph of $G$ that forms an arborescence $T_{opt} = T(r,X)$ rooted at $r$ containing all the terminals and minimizing the cost $c(T(r,X)) = \sum_{e \in T(r,X)} c(e)$.
Let $N=|X|$ denote the number of terminals and $n$ the number of vertices.

Observe that one can model {\setcover} as a special case of {\dst} on a 3-level acyclic digraph, with a universal root on top, nodes representing sets as internal layer, and the elements as leaves. The cost of an edge coming into a node corresponds to the cost of the corresponding element or set.

Our algorithm consists of "guessing" a set $C$ of intermediate nodes of the optimal tree. After computing the optimal tree on top of this set, we use this set as the source of roots for a collection of trees to cover the terminals. This becomes a set cover problem, where we map each set selected to a tree of restricted size with a root in $C$. Our algorithm then reduces to applying the classic greedy set cover algorithm on this instance induced by the "right" set $C$. 
Because of the size restriction, the resulting approximation has a smaller constant factor.

% Tree structure
We may assume each terminal is a leaf, by adding a leaf as a child of a non-leaf terminal and transfer the terminal function to that leaf.
If a tree contains a terminal, we say that the tree \emph{covers} the terminal.
%We may assume that every terminals is a leaf,and every leaf is a terminal, by simply attaching a leaf terminal to any terminal that is not a leaf and removing non-terminal leaves. Let $L$ denote the set of leaves.

Let $\ell(T)$ denote the number of terminals in a tree $T$.
Let $T_v$ denote the subtree of tree $T$ rooted at node $v$.
For node $v$ and child $w$ of $v$, let $T_{vw}$ be the subtree of $T$ formed by $T_w \cup \{vw\}$, i.e., consisting of the subtree of $T$ rooted at $w$ along with the edge to $w$'s parent ($v$).

\begin{definition}
A set $C \subset V$ is a \emph{$\phi$-core} of a tree $T$ if
there is a collection of edge-disjoint subtrees $T_1, T_2, \ldots$ of $T$ such that: % a) each $T_i$ is a subgraph of $G$, 
a) the root of each tree $T_i$ is in $C$, 
b) every terminal in $T$ is contained in exactly one tree $T_i$, and c) each $T_i$ contains at most $\phi$ terminals, $\ell(T_i) \le \phi$. 
\label{D:core}
\end{definition}

\begin{lemma}
Every tree $T$ contains a $\phi$-core of size at most $\lceil \ell(T)/\phi \rceil$, for any $\phi$. %$1 \le \phi \le n$.
\label{L:core}
\end{lemma}

\begin{proof}
The proof is by induction on the number of terminals in the tree. The root is a core when $\ell(T) \le \phi$.  
Let $v$ be a vertex with $\ell(T_v) > \phi$ but whose children fail that inequality. Let $C'$ be a $\phi$-core of $T'= T \setminus T_v$ promised by the induction hypothesis, and let $C = C' \cup \{v\}$.

For each child $w$ of $v$, the subtree $T_{vw}$ contains at most $\phi$ terminals. Together they cover uniquely the terminals in $T_v$, and satisfy the other requirements of the definition of a $\phi$-core for $T_v$. Thus $C$ is a $\phi$-core for $T$. 
Since $\ell(T') \le \ell(T) - \phi$, the size of $C$ satisfies $|C| = |C'| + 1 \le \lceil \ell(T')/\phi\rceil +1 \le
\lceil (\ell(T)-\phi)/\phi \rceil +1 = \lceil \ell(T)/\phi \rceil$.
\end{proof}

A core implicitly suggests a set cover instance, with sets of size at most $\phi$, formed by the terminals contained in each of the edge-disjoint subtrees. Our algorithm is essentially based on running a greedy set cover algorithm on that instance.
% A core $C$ splits a tree into a collection of subtrees: one connecting the root to the nodes of $C$, and then the edge-disjoint subtrees connecting the terminals to nodes in $C$. The latter collection of subtrees 

Let $R(v)$ denote the set of nodes reachable from $v$ in $G$.

\begin{definition}
%Let $G=(V,E,X,w,r)$ be a weighted digraph, 
Let $S \subseteq V$, $U \subset X$, and let $\phi$ be a parameter.
Then $S$ induces a \emph{$\phi$-bounded {\setcover} instance} $(U,\cC_{S,U}^\phi)$ with $\cC_{S,U}^\phi = \{Y \subseteq U : |Y| \le \phi \text{ and } \exists v \in S, Y \subseteq R(v) \}$.
Namely, a subset $Y$ of at most $\phi$ terminals in $U$ is in $\cC_{S,U}^\phi$ iff there is a $v$-rooted subtree containing $Y$.
\end{definition}

% Our algorithm consist of "guessing" the right $n^\alpha$-core $C$, by trying all possible subsets. We then find an optimal $r$-rooted tree with $C$ as leaves. Finally, we greedily compute subtrees with roots in $C$ to cover all the leaves of $T$.

We relate set cover solutions of $\cC_{C,X}$ to {\dst} solutions of $G$ with the following lemmas.

% DST of G gives a set cover of C_G (of same value), and a set cover of C_G gives a DST (of possibly lower value).
For a node $r_0$ and set $F$, let $T(r_0,F)$ be the tree of minimum cost that is rooted by $r_0$ and contains all the nodes of $F$.
For sets $F$ and $S$, let $T(S,F)$ be the tree of minimum cost that contains all the nodes of $F$ and is rooted by some node in $S$.

\begin{lemma}
Let $\cS$ be a set cover of $\cC_{S,X}^\phi$ of cost $c(\cS)$.
Then, we can form a valid {\dst} solution $T_{\cS}$ by combining $T(r,S)$ with the trees $T(S,F)$, for each $F \in \cS$. The cost of $T_{\cS}$ is at most $c(T_{\cS}) \le c(T(r,S)) + c(\cS)$.
\label{L:cg-sc}
\end{lemma}
\begin{proof}
$T_{\cS}$ contains all the terminals since the sets in $\cS$ cover $X$. It contains an $r$-rooted arborescence since  $T(r,S)$ contains a path from $r$ to all nodes in $S$, and the other subtrees contain a path to each terminal from some node in $S$. 
The cost bound follows from the definition of the weights of sets in $\cC_{S,X}^\phi$. The actual cost could be less, if the trees share edges or have multiple paths, in which case some superfluous edges can be shed.
\end{proof}

% DST of G gives a set cover of C_G (of same value), and a set cover of C_G gives a DST (of possibly lower value).
Let $OPT_{SC}(\cC)$ be the weight of an optimal set cover of a set system $\cC$.
\begin{lemma}
Let $C$ be a $\phi$-core of $T_{opt}$.
The cost of an optimal {\dst} of $G$ equals $OPT_{SC}(\cC_{C,X}^\phi)$ plus the cost of an optimal $r$-rooted tree with $C$ as terminals: $c(T_{Opt}) = OPT_{SC}(\cC_{C,X}^\phi) + c(T(r,C))$ 
\label{L:optcore}
\end{lemma}
\begin{proof}
The subtree of $T_{Opt}$ induced by $C$ and the root $r$ has cost $c(T(r,C))$. The rest of the tree consists of the subtrees $T_{vw}$, for each $v\in C$ and child $w$ of $v$.
$T_{vw}$ contains at most $\phi$ terminals, so the corresponding set is contained in $\cC_{C,X}^\phi$. Together, these subtrees contain all the terminals, so the corresponding set collection covers $\cC_{C,X}^\phi$. Thus, $c(T_{Opt}) \ge c(T(r,C)) + OPT_{SC}(\cC_{C,X})^\phi$. By Lemma \ref{L:cg-sc}, the inequality is tight.
\end{proof}

%The \emph{density} of a tree $T'$ is defined as $c(T')/\ell(T')$. 
The \emph{density} of a set $F$ in $\cC{S,X}^\phi$ is $\min_{s \in S} c(T(s,F))/|F|$: the cost of the optimal tree containing $F$ averaged over the nodes in $F$.

Given a root and a fixed set $S$ of nodes as leaves, an optimal cost tree $T(r,S)$ can be computed in time $poly(n) 2^{|S|}$ by a (non-trivial) algorithm of Dreyfus and Wagner \cite{dw1}.

% min-density set can be found quickly
\begin{lemma}
A minimum density set in $\cC_{S,X}^\phi$ can be found in time $n^{O(\max(\phi,N/\phi))}$.
\label{L:minset}
\end{lemma}
\begin{proof}
There are at most $2n^{\phi}$ subsets of at most $\phi$ terminals and at most $N/\phi$ choices for a root from the set $C$. Given a potential root $r_0$ and candidate core $S$, the algorithm of \cite{dw1} computes $T(r_0,S)$ in time $poly(n)2^{2N/\phi}$. 
\end{proof}

Our algorithm for {\dst} is based on guessing the right $\phi$-core $C$, and then computing a greedy set cover of $\cC_{S,X}^\phi$ by repeatedly applying Lemma \ref{L:minset}.
More precisely, we try all possible subsets $S \subset V$ of size at most $2N/\phi$ as a $\phi$-core (of $T_{opt}$) and for each such set do the following. Set $U$ initially as $X$, representing the uncovered terminals.
Find a min-density set $Z$ of $\cC_{S,U}^\phi$ and a corresponding optimal cost tree (with some root in $S$), remove $Z$ from $U$ and repeat until $U$ is empty. We then compute $T(r,S)$ and combine it with all the computed subtrees into a single tree $T_S$. The solution output, $T_{Alg}$, is the $T_S$ of smallest total cost, over all the candidate cores $S$.

\begin{theorem}
Let $\gamma \ge 1/2$ be a parameter, $\phi = N^{1-\gamma}$, and let $C$ be a $\phi$-core of $T_{opt}$. Then the greedy set cover algorithm applied to $\cC_{C,X}^\phi$ yields a $1+\ln \phi$-approximation of {\dst}. Namely, our algorithm is a $(1-\gamma)\ln n$-approximation of {\dst}.
The running time is $n^{O(max(\phi, N/\phi))} = exp(\tilde{O}(N^\gamma))$.
\label{T:dst-appx}
\end{theorem}

\begin{proof}
%Let $Alg$ be the cost of the DST tree of $G$ computed by our algorithm and recall $OPT$ is the cost of the optimal DST tree of $G$.
Let $Gr$ be the size of the greedy set cover of $\cC_{C,X}^\phi$ and $O = OPT_{SC}(\cC_{C,X}^\phi)$.
Since the cardinality of the largest set in $\cC_{C,X}^\phi$ is at most $\phi$, it follows by the analysis of Chv\'atal \cite{chva} that $Gr \le (1+\ln \phi)Opt_{SC}(\cC_{C,X}^\phi)$. Thus, 
letting $t_c = c(T(r,C))$,
\[ c(T_{Alg}) \le t_C + Gr \le t_C + (1+\ln \phi)O \le (1+\ln \phi)(t_C+O) = (1+\ln \phi) c(T_{Opt}) \ . \]
applying Lemma \ref{L:cg-sc} in the first inequality and Lemma \ref{L:optcore} in the (final) equality.
Observe that $\ln \phi = (1-\gamma)\ln n$.
For each candidate core $S$ we find a min-density set at most $n$ times. There are $\binom{n}{N/\phi} \le n^{N/\phi}$ candidate cores and the cost for each is $n \cdot n^{O(\min(\phi,N/\phi))}$, by Lemma \ref{L:minset}. Hence, the total cost is $n^{N/\phi} \cdot n^{O(\max(\phi,N/\phi))} = n^{O(N/\phi)} = exp(\tilde{O}(N^\gamma))$ using that $\phi = N^{1-\gamma} \le N/\phi$. 
\end{proof}

Now we observe that the same theorem applies to the \textsc{Connected Polymatroid} problem.
Since the function is both submodular and {\em increasing}, for every collection of pairwise disjoint sets $\{S_i\}_{i=1^k}$, it holds that $\sum_{i=1}^kf(S_i)\geq f(\bigcup_{i=1}^k S_i)$.
Thus, for a given $\gamma \ge 1/2$, 
at iteration $i$ there exists a collection $S_i$ of terminals so that $f(S_i)/c(S_i)\geq f(U)/c(U)$. We can guess $S_i$ in time $\exp(N^\gamma\cdot \log n)$ and its set of Steiner vertices $X_i$ in time $O(3^{N^{\gamma}})$.
Using the algorithm of \cite{dw1}, we can find a tree of density at most $opt/N^{\gamma}$. The rest of the proof is identical.

%\subsection*{Acknowledgments}

\bibliographystyle{plain}
\bibliography{u1}

%\appendix

\end{document}